\PassOptionsToPackage{dvipsnames}{xcolor}
\documentclass[copyright,creativecommons,submission]{eptcs}
\providecommand{\event}{} 

\usepackage{booktabs}   
\usepackage{subcaption} 
\usepackage[utf8]{inputenc}
\usepackage[english]{babel}
\usepackage[T1]{fontenc}
\usepackage{amsmath}
\usepackage{hyperref}
\usepackage{bussproofs}
\usepackage{tikz}
\usepackage{amsthm}
\usepackage{amssymb}
\usepackage{mathtools}
\usepackage{stmaryrd}

\newenvironment{bprooftree}
  {\leavevmode\hbox\bgroup}
  {\DisplayProof\egroup}

\DeclareRobustCommand{\amalg}{\mathop{\text{\fakecoprod}}}
\newcommand{\fakecoprod}{%
  \sbox0{$\prod$}%
  \smash{\raisebox{\dimexpr.9625\depth-\dp0}{\scalebox{1}[-1]{$\prod$}}}%
  \vphantom{$\prod$}%
}

\newtheorem{theorem}{Theorem}
\newtheorem{lemma}[theorem]{Lemma}
\newtheorem{proposition}[theorem]{Proposition}

\newtheorem{definition}[theorem]{Definition}
\newtheorem{assumption}[theorem]{Assumption}
\newtheorem{example}[theorem]{Example}

\newtheorem{remark}[theorem]{Remark}

\newcommand{\secref}[1]{\S\ref{#1}}

\newcommand{\Ob}{\text{Ob}}%

\newcommand{\lnl}{\ensuremath{\mathbf{\lambda}_l}}
\newcommand{\anl}{\ensuremath{\mathbf{\lambda}_a}}

\newcommand{\id}{\text{id}}
\newcommand{\Id}{\text{Id}}

\newcommand{\op}{\ensuremath{\mathrm{op}}}

\newcommand{\VV}{\ensuremath{\mathbf{V}}}

\newcommand{\BB}{\ensuremath{\mathbf{B}}}
\newcommand{\BBB}{\ensuremath{\mathbf{B}}}

\newcommand{\CC}{\ensuremath{\mathbf{C}}}
\newcommand{\CCC}{\ensuremath{\mathbf{C}}}
\newcommand{\VVV}{\ensuremath{\mathbf{V}}}

\newcommand{\DCPO}{\ensuremath{\mathbf{CPO}}}

\newcommand{\dcpo}{\DCPO}
\newcommand{\dcpobs}{\ensuremath{\mathbf{CPO}_{\perp!}}}

\newcommand{\cpobs}{\dcpobs}
\newcommand{\cpo}{\dcpo}

\newcommand{\M}{\ensuremath{\mathbf{M}}}

\newcommand{\Set}{\ensuremath{\mathbf{Set}}}

\newcommand{\lift}{\textnormal{\texttt{lift}}}
\newcommand{\lleft}{\textnormal{\texttt{left}}}
\newcommand{\rright}{\textnormal{\texttt{right}}}
\newcommand{\force}{\textnormal{\texttt{force}}}
\newcommand{\ccase}{\textnormal{\texttt{case}}}
\newcommand{\llet}{\textnormal{\texttt{let}}}

\newcommand{\eval}{\mathrm{eval}}
\newcommand{\rec}{\textnormal{\texttt{rec}}}

\newcommand{\lrb}[1]{{\llbracket #1 \rrbracket}}
\newcommand{\lrbv}[1]{{\llbracket #1 \rrbracket}_{\VV}}
\newcommand{\lrbb}[1]{{\llbracket #1 \rrbracket}_{\BB}}

\newcommand{\Wstar}{\ensuremath{\mathbf{W}^*_{\mathrm{NCPSU}}}}
\newcommand{\WNMIU}{\ensuremath{\mathbf{W}^*_{\mathrm{NMIU}}}}

\newcommand{\naturalto}{\ensuremath{\Rightarrow}}

\usetikzlibrary{decorations.pathmorphing}
\usetikzlibrary{decorations.markings}
\usetikzlibrary{decorations.pathreplacing}
\usetikzlibrary{arrows}
\usetikzlibrary{shapes}

\pgfdeclarelayer{edgelayer}
\pgfdeclarelayer{nodelayer}
\pgfsetlayers{edgelayer,nodelayer,main}

\tikzstyle{braceedge}=[decorate,decoration={brace,amplitude=10pt}]
\tikzstyle{square box}=[rectangle,fill=white,draw=black,minimum height=6mm,minimum width=6mm,yshift=0.7mm]
\tikzstyle{wire label}=[font=\footnotesize, auto,swap]

\tikzstyle{none}=[inner sep=0pt]
\tikzstyle{gn}=[circle,fill=Lime,draw=Black,line width=0.8 pt]
\tikzstyle{rn}=[circle,fill=Red,draw=Black, line width=0.8 pt]
\tikzstyle{H}=[rectangle,fill=Yellow,draw=Black]
\tikzstyle{line}=[scalar,fill=White,draw=Black]
\tikzstyle{io}=[rectangle,fill=White,draw=Black]
\tikzstyle{block}=[rectangle,fill=Orange,draw=Black]
\tikzstyle{graph}=[circle,fill=White,draw=Black]
\tikzstyle{empty}=[rectangle,fill=none,draw=none]
\tikzstyle{scaled}=[rectangle,fill=none,draw=none, font=\small]
\tikzstyle{box}=[rectangle,fill=White,draw=Black]
\tikzstyle{dot}=[circle,fill=Black,draw=Black,inner sep=0pt,minimum size=1pt]
\tikzstyle{small dot}=[circle,fill=Black,draw=Black,inner sep=0pt,minimum size=1pt]
\tikzstyle{Dot}=[circle,fill=Black,draw=Black,inner sep=0pt,minimum size=3pt]
\tikzstyle{diam}=[rectangle,fill=Black,draw,yscale=1.2,rotate=45]
\tikzstyle{gangle}=[rectangle,fill=Lime,draw=Black]
\tikzstyle{rangle}=[rectangle,fill=Red,draw=Black]
\tikzstyle{circ}=[circle,fill=none,draw=Black,scale=1.3]
\tikzstyle{ellip}=[ellipse,fill=none,draw=Black,scale=1.3,minimum width =1.3cm]
\tikzstyle{ellip2}=[ellipse,fill=White,draw=Black,scale=1.3,minimum width =3cm]
\tikzstyle{bbox}=[rectangle,fill=Blue,draw=Blue,scale=0.6]
\tikzstyle{gg}=[shape=rectangle,fill=White,draw=Black,dashed]

\tikzstyle{white circle}=[circle,fill=none,draw=Black,scale=1]
\tikzstyle{black circle}=[circle,fill=Black,draw=Black,scale=1]
\tikzstyle{grey circle}=[circle,fill=Gray,draw=Black,scale=1]
\tikzstyle{white rectangle}=[rectangle,fill=none,draw=Black,scale=1]

\tikzstyle{nodev}=[circle,fill=none,draw=Black,scale=1]
\tikzstyle{greynode}=[circle,fill=Grey,draw=Black,scale=1]
\tikzstyle{blacknode}=[circle,fill=Black,draw=Black,scale=1]
\tikzstyle{wirev}=[circle,fill=Black,draw=Black,inner sep=0pt,minimum size=3pt]
\tikzstyle{wirevred}=[circle,fill=Red,draw=Black,inner sep=0pt,minimum size=3pt]

\tikzstyle{simple}=[-,draw=Black]
\tikzstyle{to}=[->,draw=Black]
\tikzstyle{naturalto}=[-{Implies},double distance=1.5pt]
\tikzstyle{bdirected}=[<->,draw=Black]
\tikzstyle{bothdirs}=[bdirected,draw=Black]
\tikzstyle{bothdirsred}=[bdirected,draw=Red]
\tikzstyle{blue}=[-,draw=Blue]
\tikzstyle{redd}=[directed,draw=Red]
\tikzstyle{redu}=[-,draw=Red]
\tikzstyle{blued}=[directed,draw=Blue]
\tikzstyle{dash}=[dashed,draw=Black]
\tikzstyle{ddash}=[->,dashed,draw=Black]
\tikzstyle{dashedd}=[->,dashed]
\tikzstyle{dashedred}=[dashed,draw=Red]

\tikzstyle{equal-arrow}=[double equal sign distance]

\tikzstyle{dotpic}=[scale=0.5]

\tikzstyle{every picture}=[baseline=-0.25em]

\newcommand{\stikz}[2][1]{\scalebox{#1}{
\InputIfFileExists{#2}{}{\input{./tikz/#2}}
}}
\newcommand{\cstikz}[2][1]{\begin{center}\stikz[#1]{#2}\end{center}}

\title{Computational Adequacy for Substructural Lambda Calculi}
\date{}
\author{Vladimir Zamdzhiev
  \institute{Universit\'e de Lorraine, CNRS, Inria, LORIA, F 54000 Nancy, France}
}
\def\titlerunning{Computational Adequacy for Substructural Lambda Calculi}
\def\authorrunning{V. Zamdzhiev}

\begin{document}
\maketitle

\begin{abstract}
Substructural type systems, such as affine (and linear) type systems, are type
systems which impose restrictions on copying (and discarding) of variables, and
they have found many applications in computer science, including quantum
programming. We describe one linear and one affine type systems and we
formulate abstract categorical models for both of them which are sound and
computationally adequate. We also show, under basic assumptions, that
interpreting lambda abstractions via a monoidal closed structure (a popular
method for linear type systems) necessarily leads to degenerate and inadequate
models for call-by-value affine type systems with recursion. In our categorical
treatment, a solution to this problem is clearly presented. Our categorical
models are more general than linear/non-linear models used to study linear
logic and we present a homogeneous categorical account of both linear and
affine type systems in a call-by-value setting. We also give examples with many
concrete models, including classical and quantum ones.
\end{abstract}

\section{Introduction}\label{sec:intro}

Linear Logic~\cite{linear-logic} is a substructural logic where the rules for
contraction and weakening are restricted. The logic has become
very influential in computer science and it has inspired the development of linear type
systems where discarding and copying of variables is restricted in accordance with the substructural rules of linear logic. Closely related to linear type systems,
\emph{affine} type systems are
substructural type systems where only the rule for contraction is restricted, but weakening is completely
unrestricted. Both linear and affine type systems have been used to design quantum programming
languages~\cite{quant-semantics,qpl-fossacs,qpl,qlc-affine}, because they
enforce compliance with the laws of quantum mechanics, where uniform copying of
quantum information is not possible~\cite{no-cloning}.

General recursion is an important computational effect for (linear/affine)
programming languages and it is especially useful in quantum programming, due
to the probabilistic nature of many quantum algorithms and protocols which have
to be repeated until the correct solution is found. When constructing
categorical models for type systems with recursion, an important property is
\emph{computational adequacy}. Computational adequacy may be understood as
formulating an equivalent purely denotational (i.e. mathematical)
characterisation within the model of the operational notion of non-termination.
That is, one should be able to determine whether a program terminates or not
just by considering the interpretation of the program within the categorical
model\footnote{This does not imply decidability of termination, because the
interpretation may not be computable.}.

In this paper, we consider two substructural type systems -- one linear and one
affine (\secref{sec:syntax}) -- and we show how we can interpret both of them
within categorical models based on a double adjunction.  We show that we can
recover the linear/non-linear models of Benton
\cite{benton-wadler,benton-small} as special cases of our models
(\secref{sec:model}). Furthermore, our treatment of both the linear and affine
fragments of the lambda calculus we study is homogeneous -- the interpretation
of both languages are essentially the same and the models for the affine
language require only a single additional axiom. We prove soundness and
computational adequacy results for our categorical models
(\secref{sec:semantics}) and we present many concrete examples, both classical
and quantum.  In our models, we do not assume monoidal closure anywhere, and we
show that, as a special case, if one wishes to use the monoidal closure of the
computational category to interpret lambda abstractions in \emph{linear lambda
calculi} then this leads to a sound and adequate semantics, but doing so for
the call-by-value \emph{affine} language necessarily leads to degenerate models
which are inadequate (\secref{sec:degenerate-semantics}).

The models that we study in this paper are inspired by the categorical models
in \cite{sv-lnl}, but there are some differences which we discuss in
\secref{sub:definition}. Furthermore, the models we consider are also related
to Moggi's computational lambda calculus \cite{moggi}, Levy's
call-by-push-value \cite{cbpv} and the enriched effect calculus \cite{eec}. See
\cite{relating-models} for a detailed analysis of the relationship between
linear lambda calculi and their categorical models.

\section{Syntax and Operational Semantics}
\label{sec:syntax}

We begin by describing the syntax of the calculi that we will study. We will
consider two substructural lambda calculi. The first one is a mixed
linear/non-linear lambda calculus which we name $\lnl$ and the second one is a
mixed affine/non-linear lambda calculus which we name $\anl$.  Figure
\ref{fig:syntax} describes the term language of both calculi and also their
types and contexts. The two calculi only differ in their formation rules, which
we will introduce shortly.  We note that $\lnl$ has been studied in
\cite{pqm-small} (excluding recursion it appears as a fragment of
Proto-Quipper) and also in \cite{eclnl,eclnl2}.

The non-linear types (ranged over by variables $P,R$) form a subset of our
types (ranged over by variables $A,B,C$). We also distinguish between
non-linear contexts (ranged over by $\Phi$) and arbitrary contexts (ranged over
by $\Gamma, \Sigma$). Non-linear contexts contain only variables of non-linear
types, whereas arbitrary contexts may contain variables of arbitrary types
(which could be linear).

In both calculi, contraction is restricted to
non-linear types only. That is, variables of non-linear type may always be
duplicated, but in general, we do not allow copying of variables of arbitrary types
(because such a type could be linear). The only difference between $\lnl$ and
$\anl$ is that in the former, weakening is restricted to non-linear types,
whereas in the latter weakening is not restricted. This means, only variables of
non-linear type may be discarded in $\lnl$, but in $\anl$ all variables are
discardable. This is enforced by presenting different term formation rules for the
two calculi (see Figure \ref{fig:term-formation}).
In Figure \ref{fig:term-formation}, we write, as usual, $\Gamma, \Sigma$ for the union
of two disjoint contexts and $\Gamma \vdash m : A$ to indicate that term $m$ is well-formed
under context $\Gamma$ and has type $A$.
The \emph{values} are special terms which reduce to themselves in the operational semantics (see Figure \ref{fig:syntax}).
A value $\Gamma \vdash v :A$ is said to be \emph{non-linear} whenever $A$ is a non-linear type and then it is easy to
see that $\Gamma$ must also be non-linear.
See \cite{eclnl2,pqm-small} for a more detailed discussion of the syntax.

\begin{figure}
\centering
\begin{tabular}{l  l  l  l}
  Variables & $x,y,z$ & & \\
	Types & $A, B, C$ &                                   ::= & $I$  | $A+B$ | $A\otimes B$ | $A \multimap B$ | $!A$ \\
	Non-linear types & $P, R$ &                           ::= & $I$  | $P+R$ | $P\otimes R$ | $!A$ \\
  Contexts & $\Gamma , \Sigma $ &                      ::= & $x_1: A_1, x_2: A_2, \ldots, x_n : A_n$\\
  Non-linear contexts & $\Phi$ &                   ::= & $x_1: P_1, x_2: P_2, \ldots, x_n : P_n$\\
  Terms & $m, n, p$ & ::= & $x$ | $*$ | $m;n$ | \lleft$_{A,B} m$ | \rright$_{A,B} m$ \\
  & & &| \ccase{} $m$ \texttt{of} $\{$\lleft{} $x\to n\ $\rright{} $y \to p\}$ \\ 
  & & &| $\langle m, n \rangle$ | \llet{} $\langle x, y \rangle = m$ \texttt{in} $n$ | $\lambda x^A.m$ | $mn$  \\
  & & &| \lift{} $m$ | \force{} $m$ | $\rec\ z^{!A}. m$ \\
  Values & $v,w$ & ::= & $x$ | $*$ | \lleft$_{A,B} v$ | \rright$_{A,B} v$ | $\langle v, w \rangle$ | $\lambda x^A.m$ | \lift{} $m$ 
\end{tabular}
\caption{Types, terms and contexts of the $\lnl$ and $\anl$ calculi.}
\label{fig:syntax}
\end{figure}
\begin{figure}[p]
{%
{%
  \[
    \begin{bprooftree}
    \AxiomC{\phantom{$\vdash$}}
    \RightLabel{(for $\lnl$)}
    \UnaryInfC{$ \Phi, x:A \vdash x: A$}
    \end{bprooftree}
    \quad
    \begin{bprooftree}
    \AxiomC{\phantom{$\vdash$}}
    \RightLabel{(for $\lnl$)}
    \UnaryInfC{$ \Phi \vdash * : I$}
    \end{bprooftree}
    \quad
    \begin{bprooftree}
    \def\ScoreOverhang{0.5pt}
    \AxiomC{$ \Phi \vdash m : A$}
    \RightLabel{(for $\lnl$)}
    \UnaryInfC{$ \Phi \vdash \lift\ m :\ !A$}
    \end{bprooftree}
  \]

  \[
    \begin{bprooftree}
    \AxiomC{\phantom{$\vdash$}}
    \RightLabel{(for $\anl$)}
    \UnaryInfC{$ \Gamma, x:A \vdash x: A$}
    \end{bprooftree}
    \quad
    \begin{bprooftree}
    \AxiomC{\phantom{$\vdash$}}
    \RightLabel{(for $\anl$)}
    \UnaryInfC{$ \Gamma \vdash * : I$}
    \end{bprooftree}
    \quad
    \begin{bprooftree}
    \def\ScoreOverhang{0.5pt}
    \AxiomC{$ \Phi \vdash m : A$}
    \RightLabel{(for $\anl$)}
    \UnaryInfC{$ \Phi, \Gamma \vdash \lift\ m :\ !A$}
    \end{bprooftree}
  \]

  \[
    \begin{bprooftree}
    \AxiomC{$ \Phi, \Gamma \vdash m : I$}
    \AxiomC{$ \Phi, \Sigma \vdash n : A$}
    \BinaryInfC{$ \Phi, \Gamma, \Sigma \vdash  m; n : A$}
    \end{bprooftree}
    \quad
    \begin{bprooftree}
    \def\ScoreOverhang{0.5pt}
    \AxiomC{$ \Gamma \vdash m :\ !A$}
    \UnaryInfC{$ \Gamma \vdash \force\ m : A$}
    \end{bprooftree}
  \]

  \[
    \quad
    \begin{bprooftree}
    \AxiomC{$ \Gamma \vdash m : A$}
    \UnaryInfC{$ \Gamma \vdash \lleft_{A,B} m : A+B$}
    \end{bprooftree}
    \quad
    \begin{bprooftree}
    \AxiomC{$ \Gamma \vdash m : B$}
    \UnaryInfC{$ \Gamma \vdash \rright_{A,B} m : A+B$}
    \end{bprooftree}
  \]

  \[
    \begin{bprooftree}
    \def\ScoreOverhang{0.5pt}
    \AxiomC{$ \Phi, \Gamma \vdash m : A+B$}
    \AxiomC{$ \Phi, \Sigma, x : A \vdash n : C$}
    \AxiomC{$ \Phi, \Sigma, y : B \vdash p : C$}
    \TrinaryInfC{$ \Phi, \Gamma, \Sigma \vdash \ccase\ m\ \texttt{of}\ \{\lleft\ x \to n\ |\ \rright\ y \to p\} : C$}
    \end{bprooftree}
  \]

  \[
    \begin{bprooftree}
    \def\ScoreOverhang{0.5pt}
    \AxiomC{$ \Phi, \Gamma \vdash m : A$}
    \AxiomC{$ \Phi, \Sigma \vdash n : B$}
    \BinaryInfC{$ \Phi, \Gamma, \Sigma \vdash \langle m, n \rangle : A \otimes B$}
    \end{bprooftree}
    \ 
    \begin{bprooftree}
    \def\ScoreOverhang{0.5pt}
    \AxiomC{$ \Phi,\Gamma \vdash m : A\otimes B$}
    \AxiomC{$ \Phi, \Sigma,x:A,y:B \vdash n : C$}
    \BinaryInfC{$ \Phi, \Gamma, \Sigma \vdash\llet\ \langle x,y \rangle=m\ \texttt{in}\ n:C$}
    \end{bprooftree}
  \]
  
  \[
    \begin{bprooftree}
    \def\ScoreOverhang{0.5pt}
    \AxiomC{$ \Gamma, x: A \vdash m : B$}
    \UnaryInfC{$ \Gamma \vdash \lambda x^A . m : A \multimap B$}
    \end{bprooftree}
    \quad
    \begin{bprooftree}
    \def\ScoreOverhang{0.5pt}
    \AxiomC{$ \Phi, \Gamma \vdash m : A \multimap B$}
    \AxiomC{$ \Phi, \Sigma \vdash n : A$}
    \BinaryInfC{$ \Phi, \Gamma, \Sigma \vdash mn : B$}
    \end{bprooftree}
  \]

  \[
    \begin{bprooftree}
    \AxiomC{$ \Phi, z:!A \vdash m: A$}
    \UnaryInfC{$ \Phi \vdash \rec\ z^{!A}. m: A$}
    \end{bprooftree}
  \]

  \[ \text{where } \Gamma \cap \Sigma = \varnothing.  \]
}%
}%
\caption{Formation rules for $\lnl$ and $\anl$ terms.}
\label{fig:term-formation}
\end{figure}

\begin{figure}[p]
\[
\begin{bprooftree}
  \def\ScoreOverhang{0.5pt}
\AxiomC{{\color{white} $\Downarrow$}}
\UnaryInfC{$x \Downarrow x$}
\end{bprooftree}
\quad
\begin{bprooftree}
  \def\ScoreOverhang{0.5pt}
\AxiomC{{\color{white} $\Downarrow$}}
\UnaryInfC{$* \Downarrow *$}
\end{bprooftree}
\quad
\begin{bprooftree}
  \def\ScoreOverhang{0.5pt}
\AxiomC{$m \Downarrow *$}
\AxiomC{$n \Downarrow v$}
\BinaryInfC{$m;n \Downarrow v$}
\end{bprooftree}
\]

\[
\begin{bprooftree}
  \def\ScoreOverhang{0.5pt}
\AxiomC{$m \Downarrow v$}
\UnaryInfC{$\lleft\ m \Downarrow \lleft\ v$}
\end{bprooftree}
\quad
\begin{bprooftree}
  \def\ScoreOverhang{0.5pt}
\AxiomC{$m \Downarrow v$}
\UnaryInfC{$\rright\ m \Downarrow \rright\ v$}
\end{bprooftree}
\]

\[
\begin{bprooftree}
  \def\ScoreOverhang{0.5pt}
\AxiomC{$m \Downarrow \lleft\ v$}
\AxiomC{$n[v/x] \Downarrow w$}
\BinaryInfC{$\ccase\ m\ \texttt{of}\ \{\lleft\ x \to n\ |\ \rright\ y \to p\} \Downarrow w$}
\end{bprooftree}
\quad
\begin{bprooftree}
  \def\ScoreOverhang{0.5pt}
\AxiomC{$m \Downarrow v$}
\AxiomC{$n \Downarrow w$}
\BinaryInfC{$\langle m, n \rangle \Downarrow \langle v, w \rangle$}
\end{bprooftree}
\]

\[
\begin{bprooftree}
  \def\ScoreOverhang{0.5pt}
\AxiomC{$m \Downarrow \rright\ v$}
\AxiomC{$p[v/y] \Downarrow w$}
\BinaryInfC{$\ccase\ m\ \texttt{of}\ \{\lleft\ x \to n\ |\ \rright\ y \to p\} \Downarrow w$}
\end{bprooftree}
\quad
\begin{bprooftree}
  \def\ScoreOverhang{0.5pt}
\AxiomC{$m \Downarrow \langle v, v' \rangle$}
\AxiomC{$n[v/x, v'/y] \Downarrow w$}
\BinaryInfC{$\llet\ \langle x, y \rangle = m\ \text{in}\ n \Downarrow w$}
\end{bprooftree}
\]

\[
\begin{bprooftree}
  \def\ScoreOverhang{0.5pt}
  \AxiomC{{\color{white} $\Downarrow$}}
\UnaryInfC{$\lambda x.m \Downarrow \lambda x.m$}
\end{bprooftree}
\quad
\begin{bprooftree}
	\def\ScoreOverhang{0.5pt}
	\AxiomC{$m \Downarrow \lambda x.m'$}
	\AxiomC{$n \Downarrow v$}
	\AxiomC{$m'[v/x]\Downarrow w$}
	\TrinaryInfC{$mn \Downarrow w$}
\end{bprooftree}
\]

\[
\begin{bprooftree}
  \def\ScoreOverhang{0.5pt}
\AxiomC{{\color{white} $\Downarrow$}}
\UnaryInfC{$\lift\ m \Downarrow \lift\ m$}
\end{bprooftree}
\quad
\begin{bprooftree}
	\def\ScoreOverhang{0.5pt}
	\AxiomC{$m \Downarrow \lift\ m'$}
	\AxiomC{$m' \Downarrow v$}
	\BinaryInfC{$\force\ m \Downarrow v$}
\end{bprooftree}
\quad
\begin{bprooftree}
\AxiomC{$m[\lift\ \rec\ z^{!A}. m\ /\ z] \Downarrow v$}
\UnaryInfC{$\rec\ z^{!A}. m \Downarrow v$}
\end{bprooftree}
\]
\caption{Operational semantics of the $\lnl$ and $\anl$ calculi.}
\label{fig:operational}
\end{figure}

The operational semantics of $\lnl$ and $\anl$ is defined in the same way and it is standard.
It is defined in terms of a big-step call-by-value reduction relation in Figure~\ref{fig:operational}.
Writing $m \Downarrow v$ should be understood as saying that term $m$ would eventually reduce to the value $v$, at which point termination occurs.
We shall also say that a term $m$ \emph{terminates}, denoted by $m \Downarrow$, whenever there exists
a value $v$, such that $m \Downarrow v.$
Because of the presense of recursion, not all terms terminate. For example, the simplest
non-terminating program of type $A$ is $\cdot \vdash \rec\ z^{!A}.\ \force\ z : A$.
As expected, our languages satisfy subject reduction, i.e., type assignment is
preserved under term evaluation.

\begin{theorem}[Subject reduction]\label{thm:subject-reduction}
If $\Gamma \vdash m :A$ and $m \Downarrow v,$ then
$ \Gamma \vdash v :A$.
\end{theorem}

\begin{assumption}
Throughout the remainder of the paper, we assume that all terms are well-formed.
\end{assumption}

\section{Categorical Models}\label{sec:model}

In this section we describe the categorical models that we will use to
interpret our substructural lambda calculi (\secref{sub:definition}).
Afterwards, we consider the relationship of our models to other models of
intuitionistic linear logic (\secref{sub:relation}), we then formulate some
additional axioms that ensure computational adequacy holds
(\secref{sub:cpo-lnl}) and we conclude the section with concrete examples
(\secref{sub:examples}).

\subsection{Definition of the Models}
\label{sub:definition}

We start with the model for $\lnl$ which serves as the basic model of our
development. All subsequent models that we will present are specific instances
of it where some additional structure is assumed. Our model is very similar to
the one studied in \cite{sv-lnl}, but in our language we allow lifting of terms
(and not just values), so we treat "!" in a more computational way by defining
it as an endofunctor on $\CC$ (which is the usual interpretation of !) instead
of as an endofunctor on $\VV$ (which is done in \cite{sv-lnl} in order to
ensure some coherence properties which we do not need in the present paper).

\begin{definition}[$\lnl$-model]
\label{def:main-model}
A \emph{(compact) $\lnl$-model} is given by the following data:
\begin{enumerate}
  \item A cartesian category $(\BBB, \times, 1)$ with finite coproducts $(\BBB, \amalg, \varnothing);$
  \item A symmetric monoidal category $(\VVV, \otimes_{\VVV}, I_V)$ with finite coproducts $(\VVV, +_{\VVV}, 0_{\VVV});$
  \item A symmetric monoidal category $(\CCC, \otimes, I)$ with finite coproducts $(\CCC, +, 0);$
  \item A pair of symmetric monoidal adjunctions $\stikz{adjunctions.tikz}.$ We shall also write $F \coloneqq LJ : \BBB \to \CCC$, $G \coloneqq KR : \CCC \to \BBB$ and $! \coloneqq FG : \CCC \to \CCC$
    and we write $\eta : \Id \naturalto GF$ and $\epsilon : ! \naturalto \Id$ for the unit and counit, respectively, of the adjunction $F \dashv G$.
  \item For every $A \in \Ob(\VVV)$, an adjunction $L \circ (- \otimes_{\VVV} A) \dashv (A \multimap -) : \CC \to \VVV,$ called \emph{currying};
  \item[(6.)] The comonad endofunctor $! : \CCC \to \CCC$ is algebraically compact in a parameterised sense:
    for every $B \in \Ob(\BBB)$, the functor $FB \otimes !(-) : \CCC \to \CCC$ has an initial algebra
    $FB \otimes !\Omega \xrightarrow{\omega} \Omega$, such that $FB \otimes !\Omega \xleftarrow{\omega^{-1}} \Omega$ is its final coalgebra.
\end{enumerate}
\end{definition}

Let us now explain how the above data will be used for the interpretation of
$\lnl$.

The category $\BBB$ is the \emph{base} category and it has sufficient
structure to interpret non-linear values. Non-linear values are always
discardable and duplicable and because of this, $\BBB$ is assumed to be a
cartesian category. Moreover, in all of our concrete models, the above adjunctions lift to $\BBB$-enriched adjunctions and $\BBB$
serves as the \emph{base} of enrichment.

The category $\VVV$ is the category in which we interpret the \emph{values} of $\lnl$, whether they are non-linear or not.
The category $\CCC$ is the category in which we interpret all terms or \emph{computations} of $\lnl$. Because the language
is call-by-value, we have that $\lambda x^A. m$ is a value for any term $m$. Condition (5.) then allows us to interpret this by currying
the interpretation of $m$. In order to interpret the $!$ which is used for promotion of terms (especially lambda abstractions), we use
condition (4.) which ensures this can be done in a coherent way, for both values and computations. Finally, condition (6.) is used to
interpret recursion.

In many concrete $\lnl$ models, the category $\VVV$ is monoidal closed (this is the case for all concrete models we present) and the next lemma shows that condition (5.)
is then automatically satisfied.

\begin{lemma}
\label{lem:kleisli-exponential}
Assume we are given the same categorical data as in Definition \ref{def:main-model} with the exception of condition (5.). Assume further $\VVV$ is monoidal closed with $(- \otimes_{\VVV} A) \dashv (A \multimap_{\VVV} -) : \VVV \to \VVV.$
It then follows condition (5.) is satisfied.
\end{lemma}
\begin{proof}
Because $L \circ (- \otimes_{\VVV} A) \dashv (A \multimap_{\VVV} -) \circ R : \CCC \to \VVV.$
\end{proof}

Next, we formulate a categorical model for $\anl$. It can be easily recovered from models of $\lnl$ with one additional assumption.

\begin{definition}[$\anl$-model]
\label{def:anl-model}
A \emph{(compact) $\anl$-model} is given by a (compact) $\lnl$-model, where the tensor unit $I_{\VVV}$ is a terminal object of $\VVV$.
\end{definition}

\subsection{Relationship to LNL Models}
\label{sub:relation}

We will compare our models to models of intuitionistic linear logic which are also known as linear/non-linear (LNL) models \cite{benton-wadler,benton-small}.
Our models are tightly related to LNL models, but there are some subtle differences that stem from the choice of how to interpret lambda abstractions.

\begin{definition}[LNL model]
A \emph{(compact) linear/non-linear} model is given by the following data:
\begin{enumerate}
  \item A cartesian category $(\BBB, \times, 1)$ with finite coproducts $(\BBB, \amalg, \varnothing);$
  \item A symmetric monoidal \emph{closed} category $(\CCC, \otimes, \multimap, I)$ with finite coproducts $(\CCC, +, 0);$
  \item A symmetric monoidal adjunction $\stikz{adjunction-lnl.tikz}.$ 
  \item[(4.)] The functor $! = FG : \CC \to \CC$ is algebraically compact in a parameterised sense (Definition \ref{def:main-model}.6).
\end{enumerate}
\end{definition}

\begin{remark}
In the original definition of LNL models, the category $\BBB$ is assumed to be
cartesian closed. However, this is not necessary for our purposes, so we omit
this from the definition.  Nevertheless, in all concrete models we consider in
this paper, the category $\BBB$ is cartesian closed.
\end{remark}

In an LNL model, the category $\CCC$ is assumed to be monoidal closed which is
used for the interpretation of lambda abstractions, whereas in our models we do
not assume monoidal closure anywhere.  The other big difference is that values
and computations are both interpreted in the same category $\CCC$ of an LNL
model. The implications of this on the semantics is discussed in
\secref{sec:degenerate-semantics}.

We will now show that the notion of $\lnl$-model is more general than that of an LNL model.

\begin{proposition}
\label{prop:compare-models}
Every (compact) LNL model $\stikz{adjunction-lnl.tikz}$ induces a (compact) $\lnl$-model given by
$ \stikz{adjunction-degenerate.tikz}, $
where $\VVV = \CCC$ and $L = R = \Id.$ Moreover, in this case, the denotational semantics in \secref{sec:semantics} collapses precisely to the denotational semantics of \cite{eclnl,eclnl2}.
\end{proposition}

Next, let us consider a $\anl$-model (and therefore also a
$\lnl$-model), which has been used to interpret the quantum lambda calculus
(without recursion) \cite{kenta-bram} and the first-order quantum programming
language QPL (which admits recursion, but not lambda abstractions)
\cite{qpl-tcs,qpl-fossacs}.

\begin{example}
Let $\Wstar$ be the category of W*-algebras and normal completely-positive subunital maps and let $\WNMIU$ be its full-on-objects subcategory of normal multiplicative involutive unital maps.
Setting $\VV = (\WNMIU)^\op$ and $\CCC = (\Wstar)^\op$, one can define a $\anl$-model
$ \stikz{quantum-adjunction.tikz} $
the details of which are described in \cite{kenta-bram}. Moreover, in this case, the category $\CCC$ is \emph{not} monoidal closed\footnote{Bert Lindenhovius. Personal Communication.}.
\end{example}

Because the category $\CCC$ is not monoidal closed, we see that we cannot interpret lambda abstractions as in an LNL model in this case.
However, our $\lnl$-model does have sufficient
structure and lambda abstractions in \cite{kenta-bram} are (concretely)
interpreted in the same way as our (abstract) formulation in
\secref{sec:semantics}.
Therefore, by interpreting linear lambda calculi within (compact) $\lnl$-models, instead of (compact) LNL models, we can discover a larger range of concrete models for these languages.

\begin{assumption}
Throughout the remainder of the paper we only consider compact models. For brevity, when we write "LNL/$\lnl$/$\anl$-model" we implicitly assume the model is also compact.
\end{assumption}

\subsection{Computationally Adequate Models}
\label{sub:cpo-lnl}

It is possible to construct sound $\lnl$-models which are not computationally adequate. Let's consider an obvious example.

\begin{example}
\label{ex:degenerate}
The $\anl$-model $\stikz{degenerate.tikz}$ is not computationally adequate.
\end{example}

Of course, this model is completely degenerate and there is no way to
distinguish between terminating and non-terminating computations within it.
Computationally adequate models are often axiomatised in domain-theoretic terms
and we shall do so as well. Let $\cpo$ be the category with objects given
by cpo's (posets which have suprema of increasing $\omega$-chains), and with
morphisms given by Scott-continuous functions (monotone functions which
preserve suprema of increasing $\omega$-chains). Let $\cpobs$ be the
subcategory of $\cpo$ consisting of \emph{pointed} cpo's (cpo's with a least
element) and \emph{strict Scott-continuous functions} (Scott-continuous
functions that preserve the least element).

\begin{definition}
\label{def:order-model}
We shall say that a $\lnl$-model ($\anl$-model) $\stikz{adjunctions.tikz}$ is \emph{order-enriched} if:
$\BB$ and $\VV$ are $\cpo$-enriched categories, $\CC$ is a $\cpobs$-enriched category, their
coproduct and monoidal structures are $\cpo$-enriched and the functors $L,R,J,K,\multimap$ are $\cpo$-enriched functors.
\end{definition}

\begin{definition}
\label{def:adequate-model}
We shall say that a $\lnl$-model $(\anl$-model) $\stikz{adjunctions.tikz}$ is \emph{adequate} if
it is order-enriched and $\id_I \neq \perp_{I,I},$ where $\perp_{A,B}$ is the least element in the hom-cpo $\CC(A,B).$
\end{definition}

In the next section, we will show that these models are true to their name.
\subsection{Concrete Models}
\label{sub:examples}

We conclude the section by considering some concrete models.

\begin{example}
\label{ex:adequate1}
The adjunctions $\stikz{dcpo1.tikz}$ form a computationally adequate $\lnl$-model, where $U$ is the forgetful functor and $(-)_\perp$ is domain-theoretic lifting (freely adding a least element).
\end{example}

The above data, in fact, determines an LNL model.

\begin{example}
The adjunctions $\stikz{dcpo2.tikz}$ form a computationally adequate $\anl$-model.
\end{example}

In the above two examples, every object has a canonical comonoid
structure and because of this, they are not truly representative models for
linear and affine calculi. Next, we consider models where this does not hold.

\begin{example}
\label{ex:adequate2}
Let $\M$ be an arbitrary symmetric monoidal category. We can see $\M$ as a $\cpo$-enriched category when equipped with the discrete order (this is the free $\cpo$-enrichment of $\M$).
Let $\M_\perp$ be the category obtained from $\M$ by freely adding a least element to each hom-cpo (this is the free $\cpobs$-enrichment of $\M$).
Writing $\VV = [\M^{\text{op}}, \dcpo]$ for the indicated $\cpo$-enriched functor category and $ \CC = [\M_\perp^{\text{op}}, \dcpobs]$ for the indicated $\cpobs$-functor category,
we get a computationally adequate $\lnl$-model (see \cite{eclnl,eclnl2} for more discussion and details). Moreover, if the tensor unit of $\M$ is also a terminal object, then this data
is a computationally adequate $\anl$-model.
\end{example}

The above example shows a concrete model that has been used to interpret Proto-Quipper-M \cite{eclnl,eclnl2,pqm-small}, a (quantum) circuit description language.
The final model we consider is also inspired by quantum programming. It is a model of Proto-Quipper-M that supports recursive types.

\begin{example}
\label{ex:adequate3}
Let $\mathbf{qCPO}$ be the category of quantum cpo's \cite{klm} and let $\mathbf{qCPO}_{\perp!}$ be the subcategory of $\mathbf{qCPO}$ of pointed objects and strict maps.
Then the model
\[ \stikz{qcpo.tikz} \]
described in \cite{klm} is a computationally adequate $\anl$-model.
\end{example}

\section{Denotational Semantics}
\label{sec:semantics}

In this section we show how to interpret our substructural lambda calculi
within the categorical models we discussed. Every type $A$ admits an interpretation
as an object $\lrbv{A} \in \Ob(\VV)$ and as an object $\lrb{A} \in \Ob(\CC)$. In addition, every non-linear type $P$
admits an interpretation $\lrbb{P} \in \Ob(\BB).$ These interpretations are defined in Figure \ref{fig:type-interpretation} by simultaneous induction on the structure of types.
The three different type interpretations are nicely related by coherent natural isomorphisms.

\begin{figure}
\begin{align*}
\lrb{I}             &= I                           & \lrbv{I}             &= I_{\VV}                       & \lrbb I             &= 1  \\
\lrb{!A}            &= !\lrb A                     & \lrbv{!A}            &= JG \lrb A                     & \lrbb {!A}          &= G \lrb A  \\
\lrb{A+B}           &= \lrb A + \lrb B             & \lrbv{A+B}           &= \lrbv A +_{\VV} \lrbv B       & \lrbb {P+R}         &= \lrbb P \amalg \lrbb R  \\
\lrb{A \otimes B}   &= \lrb A \otimes \lrb B       & \lrbv{A \otimes B}   &= \lrbv A \otimes_{\VV} \lrbv B & \lrbb {P \otimes R} &= \lrbb P \times \lrbb R  \\
\lrb{A \multimap B} &= L(\lrbv A \multimap \lrb B) & \lrbv{A \multimap B} &= \lrbv A \multimap \lrb B      &  & 
\end{align*}
\caption{Interpretation of types.}
\label{fig:type-interpretation}
\end{figure}

\begin{proposition}
\label{prop:types}
For every type $A: \lrb{A} \cong L\lrbv{A}$. For every non-linear type $P: \lrbv P \cong J \lrbb P$ and so $\lrb P \cong L\lrbv P \cong F \lrbb P$.
\end{proposition}

These isomorphisms are also defined by induction on the structure of types, but we omit the details here (the construction is similar to the one in \cite{lnl-fpc,lnl-fpc-lmcs}).
Moreover, in order to avoid using excessive notation in the interpretation of terms, we make the following assumption.

\begin{assumption}
From now on, we suppress the natural isomorphisms related to the monoidal structure of all categories, the strong monoidal functors of the adjunction and the preservation of colimits.
With this in place, the isomorphisms of Proposition \ref{prop:types} become equalities and we will write them as such.
\end{assumption}

Of course, our results continue to hold even without this assumption, but we do this for brevity of the presentation (see \cite{lnl-fpc,lnl-fpc-lmcs} for more information on how to handle such isomorphisms).

The interpretation of a context $\Gamma = \{ x_1 : A_1, \ldots , x_n : A_n \}$ within $\CC$ is defined in the usual way as $\lrb \Gamma = \lrb{A_1} \otimes \cdots \otimes \lrb{A_n}$.
Similarly, we may define its interpretation in $\VV$. Every non-linear context $\Phi = \{ x_1 : P_1, \ldots , x_n : P_n \}$ also admits an interpretation within $\BB$ by
$\lrbb \Phi =\lrbb{P_1} \times \cdots \times \lrbb{P_n} $.
Then, just as in Proposition \ref{prop:types}, we have $\lrb \Gamma = L \lrbv \Gamma$ and for non-linear types contexts $\Phi$ we also have $\lrb \Phi = L \lrbv \Phi = F \lrbb \Phi.$

Before we may define the interpretation of terms, we have to explain how to construct morphisms for copying, deletion and promotion of non-linear primitives. We do this in the following way.

\begin{definition}\label{def:substructural-maps}
For every non-linear type or context $X$, we define discarding ($\diamond$), copying ($\triangle$) and promotion ($\Box$) morphisms in all three categories:
\begin{align*}
  \diamond_{X}^\BB  &\coloneqq \lrbb{X} \xrightarrow{1} 1  & \diamond_X^\VV &\coloneqq J \diamond_X^\BB  & \diamond_X^\CC &\coloneqq F \diamond_X^\BB \\
  \triangle_{X}^\BB &\coloneqq \lrbb{X} \xrightarrow{\langle \id, \id \rangle }  \lrbb{X} \times \lrbb{X} & \triangle_X^\VV &\coloneqq J \triangle_X^\BB & \triangle_X^\CC &\coloneqq F \triangle_X^\BB \\
  \Box_X^\BB        &\coloneqq \lrbb{X} \xrightarrow{\eta} GF \lrbb X = G \lrb X = \lrbb{!X} & \Box_X^\VV &\coloneqq J \Box_X^\BB & \Box_X^\CC &\coloneqq F \Box_X^\BB
\end{align*}
The substructural morphisms $\chi^\CC_X$ are the ones directly used for the interpretation of terms, so we shall simply write them as $\diamond_X : \lrb X \to I$ and $\triangle_X : \lrb X \to \lrb X \otimes \lrb X$ and $\Box_X : \lrb X \to \lrb{!X},$ omitting the superscript.
\end{definition}

\begin{proposition}
\label{prop:comonoids}
For every non-linear type or context $X$, the substructural maps for copying and discarding form cocommutative comonoids in their respective categories:
\begin{enumerate}
  \item The triple $(\lrbb X, \triangle_X^\BB, \diamond_X^\BB)$ is a cocommutative comonoid in $\BB$.
  \item The triple $(\lrbv X, \triangle_X^\VV, \diamond_X^\VV)$ is a cocommutative comonoid in $\VV$.
  \item The triple $(\lrb X, \triangle_X \diamond_X)$ is a cocommutative comonoid in $\CC$.
\end{enumerate}
Moreover, the comonoid homomorphisms with respect to the above structures are:
\begin{enumerate}
  \item Every morphism of $\BB$ (because $\BB$ is cartesian).
  \item The morphisms of $\VV$ in the image of $J$.
  \item The morphisms of $\CC$ in the image of $F$.
\end{enumerate}
\end{proposition}

So, we see that in any $\lnl$-model, we may define copy and discarding morphisms at every non-linear type. However, to interpret $\anl$, we have to able to construct discarding morphisms at all types (including linear ones). This is possible in a $\anl$-model, because
of the additional assumption that $I_\VV$ is a terminal object in $\VV$. Therefore, in an $\anl$-model, we simply define the discarding map to be the unique map $\diamond_A^\VV : \lrbv A \to I_\VV$, which then induces a discarding map
$\diamond_A \coloneqq L\diamond_A^\VV : \lrb A \to I$ in $\CC$. Note that the latter map can then discard any morphism in the image of $L$, and we will see that the interpretation of values satisfies this.

\begin{figure}[t]
\cstikz{alg-compactness1.tikz}
\caption{Definition of $\sigma_m$ and $\gamma_{FX}$. Given an object $FX$ and a morphism $m$ as above, $\sigma_m$ and $\gamma_{FX}$ are the unique maps making the above diagram commute, where $\Omega_{FX}$ is the initial (final) $FX \otimes !(-)$-(co)algebra.}
\label{fig:recursion}
\end{figure}

\begin{figure}
\begin{align*}
&\lrb{\Phi, x: A \vdash x : A} \coloneqq
\lrb{\Phi} \otimes \lrb{A} \xrightarrow{\diamond \otimes \id} I \otimes \lrb{A} = \lrb{A} \\
&\lrb{\Phi \vdash * : I} \coloneqq
\lrb{\Phi} \xrightarrow{\diamond} I = \lrb I \\
&\lrb{\Phi, \Gamma, \Sigma \vdash m;n : A} \coloneqq
\lrb \Phi \otimes \lrb{\Gamma} \otimes \lrb{\Sigma} 
\xrightarrow{\triangle  \otimes \id}
\lrb \Phi \otimes \lrb \Phi \otimes \lrb{\Gamma} \otimes \lrb{\Sigma}
\xrightarrow \cong \\
& \quad
\lrb \Phi \otimes \lrb{\Gamma} \otimes \lrb \Phi \otimes \lrb{\Sigma}
\xrightarrow{\lrb m \otimes \lrb n}
I \otimes \lrb A = \lrb A \\
&\lrb{\Gamma \vdash \lleft_{A,B} m : A+B} \coloneqq
\lrb{\Gamma}\xrightarrow{\lrb{m}} \lrb{A} \xrightarrow{\mathrm{left}} \lrb{A}+\lrb{B}=\lrb{A+B} \\
&\lrb{\Gamma \vdash \rright_{A,B} m : A+B} \coloneqq
\lrb{\Gamma}\xrightarrow{\lrb{m}} \lrb{B} \xrightarrow{\mathrm{right}} \lrb{A}+\lrb{B}=\lrb{A+B} \\
&\lrb{\Phi, \Gamma, \Sigma \vdash \ccase\ m\ \texttt{of}\ \{\lleft\ x \to n\ |\ \rright\ y \to p\} : C} \coloneqq
\lrb \Phi \otimes \lrb{\Gamma} \otimes \lrb{\Sigma} 
\xrightarrow{\triangle  \otimes \id} \\
& \quad
\lrb \Phi \otimes \lrb \Phi \otimes \lrb{\Gamma} \otimes \lrb{\Sigma}
\xrightarrow \cong
\lrb \Phi \otimes \lrb{\Sigma} \otimes \lrb \Phi \otimes \lrb{\Gamma}
\xrightarrow{\id \otimes \lrb m}
\lrb \Phi \otimes \lrb{\Sigma} \otimes \lrb{A+B}
\xrightarrow \cong \\
& \quad \left( \lrb \Phi \otimes \lrb{\Sigma} \otimes \lrb A \right)
+
\left( \lrb \Phi \otimes \lrb{\Sigma} \otimes \lrb B \right)
\xrightarrow{\left[ \lrb n, \lrb p \right]}
\lrb C \\
&\lrb{\Phi, \Gamma, \Sigma \vdash \langle m,n \rangle: A \otimes B} \coloneqq
      \lrb{\Phi}\otimes\lrb{\Gamma}\otimes\lrb{\Sigma}
        \xrightarrow{\triangle  \otimes \id}
      \lrb{\Phi} \otimes \lrb{\Phi}\otimes\lrb{\Gamma}\otimes\lrb{\Sigma} \xrightarrow \cong\\
     &\quad \lrb{\Phi}\otimes\lrb{\Gamma}\otimes\lrb{\Phi}\otimes\lrb{\Sigma}
      \xrightarrow{\lrb m \otimes \lrb n}
        \lrb A \otimes \lrb B = \lrb{A \otimes B} \\
&\lrb{\Phi, \Gamma, \Sigma \vdash\llet\ \langle x,y \rangle=m\ \text{in}\ n:C} := \lrb{\Phi}\otimes\lrb{\Gamma}\otimes\lrb{\Sigma}
  			\xrightarrow{\triangle \otimes\id}\lrb{\Phi}\otimes\lrb{\Phi}\otimes \lrb{\Gamma}\otimes\lrb{\Sigma}\xrightarrow \cong \\
        &\quad \lrb{\Phi}\otimes\lrb{\Gamma}\otimes\lrb{\Phi}\otimes\lrb{\Sigma}
  			\xrightarrow{\lrb{m}\otimes \id}
  			\lrb{A\otimes B}\otimes \lrb{\Phi}\otimes\lrb{\Sigma} \xrightarrow \cong
       	\lrb{\Phi}\otimes\lrb{\Sigma}\otimes\lrb{A}\otimes\lrb{B}
  			\xrightarrow{\lrb{n}}
  			\lrb{C}\\
&\lrb{\Gamma \vdash \lambda x^A . m : A \multimap B} \coloneqq
\lrb{\Gamma} = L \lrbv \Gamma \xrightarrow{L \textbf{curry}({\lrb m})} L (\lrbv A \multimap \lrb B) = \lrb{A \multimap B} \\
&\lrb{\Phi, \Gamma, \Sigma \vdash mn : B} \coloneqq
\lrb{\Phi} \otimes \lrb{\Gamma} \otimes \lrb{\Sigma}
\xrightarrow{\triangle \otimes\id}
\lrb{\Phi}\otimes\lrb{\Phi}\otimes
\lrb{\Gamma}\otimes\lrb{\Sigma}
\xrightarrow \cong \\
& \quad \lrb{\Phi}\otimes\lrb{\Gamma}\otimes\lrb{\Phi}\otimes\lrb{\Sigma}
\xrightarrow{\lrb{m}\otimes \lrb{n}}
\left( \lrb{A \multimap B} \right) \otimes \lrb{A} =
L \left( (\lrbv A \multimap \lrb B)  \otimes_\VV \lrbv{A} \right)\xrightarrow{\eval}\lrb{B} \\
&\lrb{\Phi \vdash \lift\ m :\ !A} \coloneqq \lrb \Phi \xrightarrow{\Box}\ !\lrb \Phi \xrightarrow{! \lrb m}\ ! \lrb A = \lrb{!A} \\
&\lrb{\Gamma \vdash \force\ m :A} \coloneqq \lrb \Gamma \xrightarrow{\lrb m}\ ! \lrb A \xrightarrow{\epsilon} \lrb A \\
&\lrb{\Phi \vdash \rec\ x^{!A}. m : A} \coloneqq \lrb \Phi \xrightarrow{\gamma_{\lrb \Phi}} \Omega_{\lrb \Phi} \xrightarrow{\sigma_{\lrb m}} \lrb A
\end{align*}
\caption{Interpretation of $\lnl$-terms.}
\label{fig:term-interpretation}
\end{figure}

We many now define the interpretation of terms of $\lnl$. As usual, a
well-formed term $\Gamma \vdash m : A$ is interpreted as a morphism $\lrb { \Gamma \vdash m : A} :
\lrb \Gamma \to \lrb A$ in $\CC$ which is defined by induction on the
derivation of $\Gamma \vdash m :A$ in Figure \ref{fig:term-interpretation}.
We will also often abbreviate this by simply writing $\lrb m$, instead of $\lrb { \Gamma \vdash m : A}$.
The interpretation of recursion makes use of the auxiliary definition in Figure \ref{fig:recursion} which is well-defined due to the assumption in Definition \ref{def:main-model}.6.
In the interpretation of $\ccase$ terms, we use the fact that the tensor product distributes over coproducts, provided they are both in the image of $L$, which follows from the assumption in Definition \ref{def:main-model}.5.

The interpretation of $\anl$ terms within a $\anl$-model is done in the same way as in Figure \ref{fig:term-interpretation}, but where we update the three rules that are different among the calculi to also handle the more general contexts, as follows:
\begin{align*}
&\lrb{\Gamma, x: A \vdash x : A} \coloneqq
\lrb{\Gamma} \otimes \lrb{A} \xrightarrow{\diamond \otimes \id} I \otimes \lrb{A} = \lrb{A} \\
&\lrb{\Gamma \vdash * : I} \coloneqq
\lrb{\Gamma} \xrightarrow{\diamond} I = \lrb I \\
&\lrb{\Phi, \Gamma \vdash \lift\ m :\ !A} \coloneqq \lrb \Phi \otimes \lrb \Gamma \xrightarrow{\id \otimes \diamond} \lrb \Phi \otimes I = \lrb \Phi \xrightarrow{\Box}\ !\lrb \Phi \xrightarrow{! \lrb m}\ ! \lrb A = \lrb{!A}
\end{align*}

In order to show that our models are sound, we have to show that the interpretations of (non-linear) values interact nicely with the substructural morphisms we have defined (Proposition \ref{prop:substructural}).
This is done by showing that non-linear values admit an interpretation in $\BB$ and that all values admit an interpretation in $\VV$, such that the interpretation of values in $\CC$ are in the image of the respective left adjoints.

\begin{lemma}
For every non-linear value $\Phi \vdash v : P$, we define an interpretation $ \lrbb{ \Phi \vdash v : P } : \lrbb{\Phi} \to \lrbb P $ within $\BB$ by induction on the derivation of $\Phi \vdash v : P$ as follows:
\begin{align*}
\lrbb{\Phi, x : P \vdash x : P} &\coloneqq \lrbb \Phi \times \lrbb P \xrightarrow{\pi_2} \lrbb P \\
\lrbb{\Phi \vdash * : I } &\coloneqq \lrbb \Phi \xrightarrow{1} 1 = \lrbb 1 \\
\lrbb{\Phi \vdash \lleft_{P,R} v : P +R } &\coloneqq \lrbb \Phi \xrightarrow{\lrbb v} \lrbb P \xrightarrow{\mathrm{inl}} \lrbb P \amalg \lrbb R = \lrbb{P + R} \\
\lrbb{\Phi \vdash \rright_{P,R} v : P +R } &\coloneqq \lrbb \Phi \xrightarrow{\lrbb v} \lrbb R \xrightarrow{\mathrm{inr}} \lrbb P \amalg \lrbb R = \lrbb{P + R} \\
\lrbb{\Phi \vdash \left\langle v, w \right\rangle : P \otimes R } &\coloneqq \lrbb \Phi \xrightarrow{ \left\langle \lrbb v, \lrbb w \right\rangle } \lrbb P \times \lrbb R = \lrbb{P \otimes R} \\
\lrbb{\Phi \vdash \lift\ m :\ !A } &\coloneqq \lrbb \Phi \xrightarrow{\eta} GF\lrbb \Phi = G \lrb \Phi \xrightarrow{G \lrb m} G \lrb A = \lrbb{!A}
\end{align*}
Then $\lrb v = F \lrbb v.$
\end{lemma}

Using the same idea, we can define for every value $v$ an interpretation $\lrbv v$ in $\VV$.

\begin{lemma}
For every value $\Gamma \vdash v : A$ of both $\lnl$ and $\anl$ it is possible to define an interpretation $\lrbv{ \Gamma \vdash v : A} : \lrbv \Gamma \to \lrbv A$ within $\VV$, such that $\lrb v = L \lrbv v$ (details ommitted for lack of space).
\end{lemma}

\begin{proposition}
\label{prop:substructural}
In any $\lnl$-model, for every non-linear value $\Phi \vdash v : P$, we have:
\begin{align*}
  \diamond_{P} \circ \lrb v &= \diamond_{\Phi} &   \triangle_{P} \circ \lrb v &= (\lrb v \otimes \lrb v) \circ \triangle_{\Phi} &   \Box_{P} \circ \lrb v &=\ ! \lrb v \circ \Box_{\Phi}.
\end{align*}
Moreover, in any $\anl$-model, for every value $\Gamma \vdash v : A,$ we also have that
  $\diamond_{A} \circ \lrb v = \diamond_{\Gamma} .$
\end{proposition}

With this place, soundness may now be proved in a straightforward way.

\begin{theorem}[Soundness]
\label{thm:soundness}
If $\Gamma \vdash m : A$ and $m \Downarrow v,$ then $\lrb m = \lrb v$ (in both $\lnl$ and $\anl$).
\end{theorem}

The above theorem shows that $\lnl$ $(\anl)$ can be soundly interpreted in any $\lnl$-model ($\anl$-model). To prove computational adequacy, we need some additional assumptions (as Example \ref{ex:degenerate} demonstrates).

\begin{theorem}[Adequacy]
\label{thm:adequacy}
Let $\cdot \vdash p : I$ be a closed program of unit type in $\lnl$ $(\anl)$. Then in any computationally adequate $\lnl$-model ($\anl$-model):
\[ \lrb p \neq \perp \text{ iff } p \Downarrow. \]
\end{theorem}
\begin{proof}
This may be established using standard proof techniques for adequacy, e.g. \cite{quant-semantics,qpl-fossacs}.
\end{proof}

\section{Lambda Abstractions, Monoidal Closure and Adequacy}
\label{sec:degenerate-semantics}

One of the stated goals of the present paper is to study how lambda
abstractions may be interpreted for substructural lambda calculi and the
effects this has on computational adequacy.  We have shown that if one uses
currying through our category $\VV$, then both $\lnl$ and $\anl$ can be
interpreted in a sound and adequate way in \secref{sec:semantics}. However, for
linear lambda calculi, ones often sees lambda abstractions interpreted
using the monoidal closed structure of the computational category. In this
section, we will assume that our category $\CC$ is monoidal closed and update
the semantics to interpret lambda abstractions through this structure and we then
show that this does not cause problems for $\lnl$, but it does cause problems for
$\anl$.

\begin{assumption}
\label{assume:last}
Throughout the remainder of the section, we assume that the category $\CC$ of a $\lnl$-model is a symmetric monoidal closed category $(\CC, \otimes, \multimap, I)$. The functor $\multimap : \CC^\op \times \CC \to \CC$ now refers to the functor induced by the adjunction $(- \otimes A) \dashv (A \multimap -)$ of
the symmetric monoidal closed structure.
\end{assumption}

As a special case of Proposition \ref{prop:compare-models}, by taking $\VV = \CC$ and $L = \Id = R$, we get a sound model of $\lnl$, where lambda abstractions are interpreted via the monoidal closed structure of $\CC$.
Under the additional assumptions of Definition \ref{def:adequate-model}, we also get computational adequacy, and it is not too difficult to find computationally adequate concrete models (Examples \ref{ex:adequate1}, \ref{ex:adequate2}, \ref{ex:adequate3}).
Therefore, we see that interpreting lambda abstractions via the monoidal closed structure of the computational category is not a problem for \emph{linear} lambda calculi.

Next, let us consider the situation for $\anl$. We will first explain how to interpret $\anl$ using the newly assumed structure.
The interpretation of types is updated by setting $\lrb{A \multimap B} = \lrb A \multimap \lrb B.$ The interpretation of lambda abstractions and application are updated as follows:
\begin{align*}
&\lrb{\Gamma \vdash \lambda x^A . m : A \multimap B} \coloneqq
\lrb{\Gamma} \xrightarrow{\textbf{curry}({\lrb m})} (\lrb A \multimap \lrb B) = \lrb{A \multimap B} \\
&\lrb{\Phi, \Gamma, \Sigma \vdash mn : B} \coloneqq
\lrb{\Phi} \otimes \lrb{\Gamma} \otimes \lrb{\Sigma}
\xrightarrow{\triangle \otimes\id}
\lrb{\Phi}\otimes\lrb{\Phi}\otimes
\lrb{\Gamma}\otimes\lrb{\Sigma}
\xrightarrow \cong \\
& \quad \lrb{\Phi}\otimes\lrb{\Gamma}\otimes\lrb{\Phi}\otimes\lrb{\Sigma}
\xrightarrow{\lrb{m}\otimes \lrb{n}}
\left( \lrb{A} \multimap \lrb{B} \right) \otimes \lrb{A}\xrightarrow{\eval}\lrb{B}
\end{align*}
We may now show the interpretation is provably inadequate and also completely degenerate.

\begin{proposition}
\label{prop:final}
Assume we are given a sound $\anl$-model under Assumption \ref{assume:last}. Then $\CC \simeq \mathbf 1$ and so the $\anl$-model is not computationally adequate (because every homset of $\CC$ has exactly one morphism).
\end{proposition}
\begin{proof}
The monoidal closure of $\CC$ together with Definition \ref{def:main-model}.6 imply that $\CC$ is a pointed category (has a zero object) \cite[Theorem 4.9]{eclnl} and so we shall write $\perp_{A,B} : A \to B$ for its zero morphisms.
The monoidal closure of $\CC$ implies that $A \otimes 0 \cong 0$, for every $A \in \Ob(\CC)$ and that $f \otimes \perp_{C,D} = \perp_{A \otimes B, C \otimes D}$, for any $f: A \to B$.

Let $p = \rec\ z^{!I}.\ \force\ z$. Then, $\cdot \vdash p : I$ and moreover $p \not \Downarrow$ with $\lrb p = \perp$ \cite[Theorem 4.9]{eclnl}.
But then
\[ \lrb {\cdot \vdash \lambda x^I. p : I \multimap I} =  \mathbf{curry}(\perp) = (I \multimap \perp) \circ \eta' = \perp \circ \eta' = \perp. \]
Next, consider the program $t = (\lambda y^{I \multimap I}. *)(\lambda x^I. p)$. This program is well-formed in $\anl$ with $\cdot \vdash t : I$ (but it is not well-formed in $\lnl$) and $t \Downarrow *$.
By soundness, $\lrb t = \lrb * = \id_I$. By definition of $\lrb{-}$, we have
\[
  \lrb t = \eval \circ (\mathbf{curry}(\diamond) \otimes \perp) \circ \cong \circ (\triangle \otimes \id) = \eval \circ \perp \circ \cong \circ (\triangle \otimes \id) = \perp.
\]
This means $\id_I$ is a zero morphism and so $I$ is a zero object. Then, every $A \in \Ob(\CC)$ is a zero object, because $A \cong A \otimes I \cong A \otimes 0 \cong 0$
and therefore $\CC \simeq \mathbf 1.$
\end{proof}

\begin{remark}
If we model recursion by assuming our model is order-enriched instead of compact in Definition \ref{def:main-model}, then one can also show that the model becomes degenerate using similar arguments.
\end{remark}

\begin{remark}
If one assumes the full law of beta-equivalence is satisfied by the language, then the degeneration can be demonstrated for a wider class of models as well.
\end{remark}

\paragraph{Acknowledgements.} I thank Bert Lindenhovius for discussions about
W*-algebras. I also thank the anonymous reviewers for their feedback and I
gratefully acknowledge financial support from the French projects
ANR-17-CE25-0009 SoftQPro and PIA-GDN/Quantex.

\bibliography{refs}
\end{document}